%% file: Flag_v18_PRXQ.tex
\documentclass[12pt,aps,prx,reprint,floatfix,longbibliography]{revtex4-2}

\usepackage{standalone}
\usepackage{import}

\usepackage{amsmath,amssymb,amsthm}
\usepackage{braket}
\usepackage{empheq}
\usepackage[dvipsnames, table]{xcolor}
\usepackage{qcircuit}
\usepackage[normalem]{ulem}
\usepackage{graphicx}
\usepackage[caption=false]{subfig}
\usepackage{mathrsfs}
\usepackage{tabularx}
\usepackage{makecell}
\graphicspath{{figures/}}

\usepackage{algorithm}
\usepackage{algpseudocode}
\newcommand{\algorithmicinput}{\textbf{input}}
\newcommand{\INPUT}{\item[\algorithmicinput]}

\usepackage{mathtools} 
\newcommand{\ceil}[1]{\left\lceil {#1} \right\rceil}

\usepackage{ stmaryrd }

\usepackage[pdfstartview=FitH]{hyperref}
\hypersetup{
    colorlinks=true,       
    linkcolor=cyan,          
    citecolor=magenta,        
    filecolor=magenta,      
    urlcolor=cyan,           
    runcolor=cyan
}

\newcolumntype{G}{>{\columncolor{Green!50}}c}
\newcolumntype{B}{>{\columncolor{Blue!40}}c}

\newcommand{\ba}{\begin{align}}
\newcommand{\ea}{\end{align}}

\newcommand\abs[1]{\left|#1\right|}

\newcommand{\Z}{\ensuremath{\mathbb{Z}}} 

\DeclareMathOperator*{\argmin}{arg\,min}

\newtheorem{theorem}{Theorem}[section]
\newtheorem*{theorem*}{Theorem}
\newtheorem{corollary}{Corollary}[theorem]
\newtheorem{lemma}{Lemma}[section]
\newtheorem{definition}{Definition}

\begin{document}

\title{Flag Gadgets based on Classical Codes}
\author{Benjamin Anker}
\affiliation{Department of Electrical and Computer Engineering}
    \affiliation{Center for Quantum Information and Control\\University of New Mexico}
\author{Milad Marvian}
\affiliation{Department of Electrical and Computer Engineering}
    \affiliation{Center for Quantum Information and Control\\University of New Mexico}

\begin{abstract}
Fault-tolerant syndrome extraction is a key ingredient in implementing fault-tolerant quantum computations. While conventional methods use a number of extra qubits linear in the weight of the syndrome, several improvements have been introduced using flag gadgets.
In this work, we develop a framework to design flag gadgets using classical codes. Using this framework we show how to perform fault-tolerant syndrome extraction for any stabilizer code with arbitrary distance using exponentially fewer qubits than conventional methods when qubit measurement and reset are relatively slow compared to a round of error correction. 
We further take advantage of the saving provided by our construction to fault-tolerantly measure multiple stabilizers using a single gadget, and show that it maintains the same exponential advantage when it is used to fault-tolerantly extract the syndrome of quantum LDPC codes. Using the developed framework we perform computer-assisted search to find several small examples where our constructions reduce the number of qubits required. These small examples may be relevant to near-term experiments on small-scale quantum computers.

\end{abstract}

\maketitle

\section{Introduction}

A key challenge in designing fault-tolerant circuits is to control the spread of faults in the computation. The spread of faults in performing certain quantum operations can be controlled by the transversal implementation of the encoded operation, i.e., performing encoded operations using parallel local operations. Incorporating this strategy to perform syndrome measurement requires using many ancilla qubits and preparing certain less noisy highly entangled quantum states, which both need to be done using faulty quantum operations. To satisfy these requirements, the traditional strategy has been to take a conservative approach to detect errors in any circuit location by introducing many ancillary qubits and gates, and discarding and restarting the fault-tolerant subroutine any time that an error is detected. Since syndrome extraction is the most frequently needed subroutine in fault-tolerant constructions, this strategy leads to significant overhead in the number of qubits and gates to make a quantum algorithm fault-tolerant. For example, Shor's fault-tolerant error correction scheme~\cite{sh,divincenzo1996fault} requires at least as many ancilla qubits as the weight of the syndrome measurement, Steane's error correction scheme~\cite{st} requires as many ancilla qubits as the size of the code block, and recent progress has shown how to interpolate between the cost of these two schemes~\cite{huang2021between,huang2021constructions}. Recently, it has been shown that  careful designing gadgets that can “flag” the location of errors and adopting a more relaxed strategy in the error correction procedure can lead to substantial savings in the overhead cost of fault-tolerant schemes~\cite{1,2}.

The flag qubit paradigm \cite{1, 3, 4, 5} 
allows low-weight errors during syndrome extraction to propagate to high-weight correlated errors on the data. However,  extra structure is added to the syndrome extraction circuit so that the low-weight errors can be identified and the correlated errors can be corrected. As originally presented \cite{1}, the flag qubit technique only applies to the extraction of a handful of different syndromes and may require adaptive measurements, where syndrome measurements depend upon the flag pattern received. Since its inception, the technique has been extended in many directions \cite{steane_flag,topological, ion_topological,cyclic,trivalent, capped,ewp,4, gkp, magic, magic2,cert,veri}, and has enabled several recent experiments \cite{realization, 5q, shuttling, diamond, 2qubit_ex}, allowing for realizations of fault-tolerance using current hardware in some cases.

In particular, flag gadgets have been shown to apply to any stabilizer code \cite{3}. In architectures that allow for relatively fast qubit measurement and reset, syndrome extraction can be performed fault-tolerantly using only a constant number of ancilla qubits \cite{3}. However, when qubit measurement and reset is slow or unavailable, this construction uses a number of ancilla qubits that scales linearly with the weight of the syndrome to be extracted, similar to conventional methods such as Shor's method \cite{sh}. 
Most relevant to this work, it has been shown that flag gadgets can be used for fault-tolerant syndrome extraction of any distance 3 code, using only a number of ancilla qubits logarithmic in the weight of the stabilizer, without requiring fast qubit measurement/reset \cite{4}. 

An open question relating to flag gadgets is whether a qubit saving over conventional methods can be realized for general codes, without restricting the form or distance, and without requiring fast measurement and reset.
We answer this open question affirmatively and constructively.  

In this work, we develop a framework to use classical codes in designing flag gadgets.
This framework protects the qubit used for syndrome extraction with the parity checks used by a classical code with the same distance as the quantum code.  
The measurement of each flag qubit used corresponds to the result of one of the parity checks in the classical code.
The resulting flag qubit measurements can then be used to locate propagating errors on the syndrome, in analogy to the way in which parity checks for a classical code locate errors on a codeword.
We show that locating propagating errors on the syndrome is sufficient to ensure fault-tolerant syndrome extraction. In addition, we show how the parity checks can be approximately implemented under physically relevant constraints such as limitations on the number of gates that can act on a qubit simultaneously, and that the effect of imprecision can be handled by repeating the parity checks in space.

Using this framework, we present a construction that can be applied to any stabilizer code, with arbitrary distance. For a distance $d=2t+1$ quantum code, this construction uses the parity check matrix for the distance $d$ BCH code, repeated in space $d$ times. The BCH code uses only $t \ceil{\log_2 (w + 1)}$ parity checks for a $w$ bit codeword, so our construction uses only $(2t+1)t \ceil{\log_2 w}$ flag qubits for extracting a weight $w$ syndrome of a distance $d$ quantum code with slow reset, which is an exponential saving comparing to the conventional methods that use $O(w)$ many qubits.

Additionally, we show that we can apply our framework to a sequence of stabilizer measurements, instead of a single stabilizer measurement. 
Because of the logarithmic scaling of the proposed method, this technique allows for constructions that use vastly fewer qubits than conventional methods, when qubit measurement and reset is relatively slow. In particular, we show that in the absence of qubit reset, our proposed scheme can extract the syndrome of any quantum Low-Density Parity-Check (qLDPC) code \cite{breuckmann2021ldpc} with exponentially fewer ancilla qubits compared to  Shor style syndrome extraction.

The mathematical framework proposed in this work provides a systematic approach to design fault-tolerant gadgets, by enabling computer-assisted search to optimize the resources. We provide many examples of optimized small gadgets potentially suitable for near-term experiments. In addition, we provide three algorithms to decode the flags, i.e., to identify the locations of errors based on the flag patterns. We also discuss the savings provided using our constructions when the qubit reset is available but is slow compared to the two-qubit gates.

Our framework focuses on protecting the measurement of one or more stabilizer measurements against propagating errors. Because of this focus, our construction can straightforwardly replace any instance of Shor syndrome extraction, without requiring modification of the rest of the error correction procedure and therefore allows our construction to introduce qubit savings for alternative fault-tolerant error correction procedures \cite{bys, encsyn, adaptive_syn}.

The organization of this paper is as follows.
In Section~\ref{sec:flag_gadgets} we outline the main ideas of the flag qubit paradigm.
In Section~\ref{sec:FT_req} we  introduce our framework to represent flag gadgets and syndrome extraction circuits using a binary matrix and state the fault-tolerance requirements in this language.
In Section~\ref{sec:ideal} we first present the intuition in the design of our framework under some simplifying assumptions on the connectivity of the device.
In Section~\ref{sec:phys_const} we remove the simplifying assumption and impose physically relevant constraints on the connectivity of the device and provide a proof that our construction is fault-tolerant.  
The cost analysis is provided in Section~\ref{sec:cost} and in Section~\ref{sec:decoding} we discuss alternative decoding schemes.
Section~\ref{sec:multi_syn} outlines a technique to apply our construction to multiple syndrome extraction circuits as if they were one, larger, syndrome extraction circuit, and Section~\ref{sec:multi_syn_resource} analyzes the resources required by this technique, as opposed to flagging each measurement separately.
In Section~\ref{sec:sims} we present the results of numerical simulations verifying fault-tolerance for certain codes,  in addition to providing examples of computer-assisted designs of syndrome measurements for small codes. 
Section~\ref{sec:reset_time} analyzes the trade-offs between our construction and other methods for fault-tolerant syndrome extraction as measurement and reset times vary.


\section{Flag Gadgets}\label{sec:flag_gadgets}
The idea of using flag gadgets is to attach a simple quantum circuit to the syndrome qubit, such that we can identify the location of low-weight errors on the syndrome and correct the propagated errors on the data qubits.

This idea is illustrated in Figure~\ref{fig:flag_intro} using a simple example to measure $X^{\otimes 3}$. Flag qubits are connected to the syndrome qubit by CNOT gates.  Given the structure of the circuit, there are two types of faults we need to consider: syndrome faults and flag faults. Syndrome faults have nontrivial support on the syndrome qubit, while flag faults have nontrivial support on a flag qubit. The errors on the syndrome qubit can spread to data qubits. Faults with trivial support on the syndrome qubit do not spread to data but can cause a ``fake" flag pattern, which can trigger a correction step that introduces more errors on the data. The flags should identify syndrome faults while remaining resistant to flag faults.

If a flag qubit produces $-1$ upon being measured, we know that an error on the syndrome qubit occurred before an odd number of CNOTs connecting the syndrome qubit to this flag qubit (or that a flag error occurred). We can then think of a pair of CNOTs connected to a single flag qubit as a ``flag'', which causes the flag qubit to switch sign whenever an error occurs on the syndrome qubit in the region it flags (the region between the two CNOTs). This is illustrated in Figure~\ref{fig:flag_intro}. 
The error on the syndrome qubit is then inferred from the pattern of $+1$ and $-1$ given by the flag qubits (i.e. from the \emph{flag pattern}).

Note that flag gadgets have been used to enable two distinct error correction strategies. Flag gadgets were first used to protect error correction using \emph{adaptive} syndrome extraction, where the choice of stabilizers measured is conditional on the results of the previous flag and stabilizer measurements \cite{1, cyclic, capped, bridge, steane_flag}. However, subsequent generalizations have removed this requirement and instead focus on inferring errors during a static set of stabilizer measurements \cite{3, 4}. It is the latter picture that we focus on.



\begin{figure}
    \includegraphics{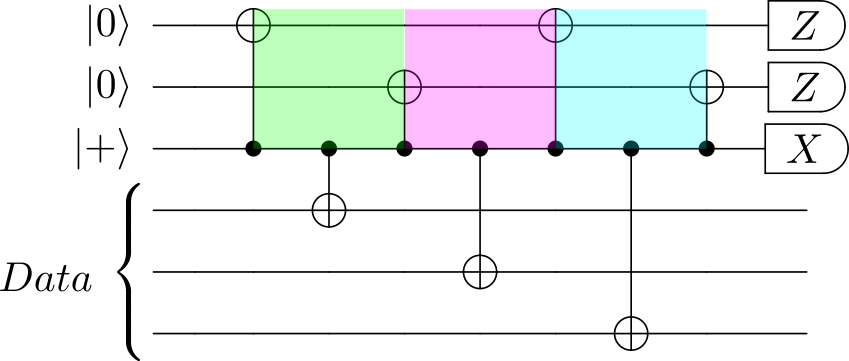}
    \caption{\label{fig:flag_intro} Illustration of a flag qubit circuit for measuring $X^{\otimes 3}$. The two qubits prepared in the $\ket 0 $ state are the flags and the qubit prepared in the $\ket +$ state is the syndrome qubit. If an error that can propagate to the data (an $X$ error on the syndrome) occurs in the green (leftmost) region, only the top flag will produce a $-1$ result when measured. If an error occurs in the violet (middle) region, both flags will produce $-1$. If an error occurs in the aqua (rightmost) region, only the second flag will produce $-1$. Therefore, considering only errors on syndrome qubit, using the flag patterns we can distinguish and correct errors that propagate to the data.}
\end{figure}

The other consideration in designing flag gadgets is the qubit reset time. Some flag gadget constructions, such as  \cite{3, 5}, assume that flag qubits can be reused within one stabilizer measurement. This requires that qubit measurement and reset are relatively fast, or requires idling while qubits are measured and reset (allowing more errors to be introduced). Our construction in general does not require qubits to be reset quickly, but may benefit from a fast reset. The resources required are analyzed for different reset timescales in Section~\ref{sec:reset_time}.

The idea of using flag gadgets has been extended in various directions. Flag gadgets have been applied to measuring the syndrome for the $7$ qubit Steane code \cite{steane_flag}, topological and subsystem codes \cite{topological, ion_topological}, cyclic CSS codes \cite{cyclic}, color codes \cite{trivalent, capped}, concatenated codes \cite{ewp}, state preparation \cite{4, gkp, magic, magic2}, entanglement certification \cite{cert}, and circuit verification \cite{veri}. 
The technique has also been applied to schemes for reducing the number of stabilizer measurements required for fault-tolerance \cite{brown}, and to measuring stabilizers for a distance 3 CSS code in parallel \cite{liou}.
Flag gadgets have been shown to be equivalent to encoding the ancilla qubit as a logical ancilla in certain cases \cite{bridge}.
Using flag gadgets to protect against hardware-specific noise models has also been explored \cite{rydberg, ion, damping, crosstalk}.

\section{Framework for Describing Flag Gadgets}\label{sec:FT_req}
In this section, we introduce notation to describe the flag gadgets using certain binary matrices and state the fault-tolerance requirement using this notation.


We assume that the stabilizer $\sigma_1 \ldots \sigma_w$ is measured by performing a controlled $\sigma_i$ gate on the $i-$th data qubit with the control prepared in the $\ket +$ state, then measuring the control (syndrome) qubit in the $X$ basis. Then the only errors on the syndrome that can propagate to the data from the syndrome qubit are $X$ errors. In examples and simulations presented in this paper, we assume that the stabilizer is of the form $X^{\otimes w}$, but the results can be directly applied to measuring general stabilizer operators. 

Flag gadgets are used to identify where on the syndrome qubit $X$ errors occurred, from which we can deduce the propagated error as a function of the form of the stabilizer.

Since we only need to identify the location of $X$ errors on the syndrome qubit, we describe the syndrome error $e_s$ by a binary column vector of height $l$ and weight $t_s$, where $ l $ is the number of possible locations for an error (which depends on the weight of the syndrome and on the flag construction).
A key ingredient in specifying the fault-tolerant property of a circuit, and also in our constructions, is a description of how error propagates from the syndrome onto flag qubits. 
\begin{definition}
Let $F$ be the matrix that describes how the syndrome error propagates to the flag qubits, i.e. $F e_s$ is a binary column vector with $(Fe_s)_i = 1$ iff an odd number of errors from the syndrome qubit propagate to flag qubit $i$.
\end{definition}
Similarly, describe the flag error $e_f$ by a binary vector column of height $f$ and weight $t_f $, corresponding to the errors occurring on the flag qubits, where $ f $ is the number of flag qubits a construction uses. Since errors on the flag qubits do not propagate to the data, only the parity of the error on a flag qubit needs to be considered, justifying the definition of $e_f$ as a binary vector. Therefore measuring the flag qubits reveals $F e_s \oplus e_f$. Based on the result of flag measurements we infer the data qubits we need to apply a correction to; we call the corresponding map $R$. 

We also define a matrix describing how errors propagate to the data.
\begin{definition}
    Define $D$ such that $(D e_s)_i = 1$ iff an odd number of syndrome errors propagated to data qubit $i$.
\end{definition}

The goal of using flag qubits is to fault-tolerantly extract a syndrome using a few qubits. In this work, we use fault-tolerance in the strong sense \cite{FTEC,adaptive_syn}.

\begin{definition}
Syndrome extraction is said to be (strongly) \emph{fault-tolerant} if the following conditions hold: 
\begin{enumerate}
    \item If the data is a weight $r$ Pauli correction from codeword $c$ before syndrome extraction and $s$ faults occur during syndrome extraction and $r + s\leq t$ after syndrome extraction the data is at most a weight $r + s$ Pauli correction from codeword $c$
    \item If the data is a weight $r$ Pauli correction from codeword $c$ before syndrome extraction and $s$ faults occur during syndrome extraction, after syndrome extraction the data is at most a weight $r + s$ Pauli correction from \emph{any} codeword $c'$.
\end{enumerate}
\end{definition}

Given these definitions, we succeed in fault-tolerantly identifying errors on the syndrome qubit if
\begin{equation}
    D e_s \oplus R(F e_s \oplus e_f)
    \label{eq:FT}
\end{equation}
has weight at most $t_f+t_s=k$, for any $k\leq t$, up to a stabilizer. 


We can describe a flag gadget circuit diagram using a $(f+1)\times l$  binary matrix $C$ as follows. The zeroth row of $C$ specifies the location of the CNOTs connecting the syndrome qubit to the data qubits. This row has $w$ many non-zero elements. The remaining $f$ rows describe the location of CNOTs connecting the syndrome qubit to flag qubits.  The matrix element $C_{i,j}=1$ iff there is a CNOT between the syndrome qubit and $i$-th flag qubit at location $j$. 
This is summarized in Figure~\ref{fig:circuit_matrix_to_circit}.

\begin{figure*}[t!]
    \centering
    {\def\quad{\hspace{1.1em}}
    \hspace{-2.9em}$\bordermatrix{
        ~ & \cr
         FLAG1&1 & 0 & 0 & 0 & 0 & 0 & 0 & 1 & 0 & 0 & 1 & 0 & 0 & 1 & 0 & 0 & 0 & 0 & 0\cr
        FLAG2&0 & 0 & 1 & 0 & 0 & 1 & 0 & 0 & 0 & 1 & 0 & 0 & 0 & 0 & 0 & 0 & 0 & 0 & 1\cr
        FLAG3&0 & 0 & 0 & 0 & 1 & 0 & 0 & 0 & 0 & 0 & 0 & 0 & 1 & 0 & 0 & 1 & 0 & 1 & 0\cr
       DATA & 0 & 1 & 0 & 1 & 0 & 0 & 1 & 0 & 1 & 0 & 0 & 1 & 0 & 0 & 1 & 0 & 1 & 0 & 0
}$ }\\
    \input{figures/hamming7norep}
    \caption{\label{fig:circuit_matrix_to_circit} Illustration of how a circuit matrix corresponds to a physical circuit.}
\end{figure*}
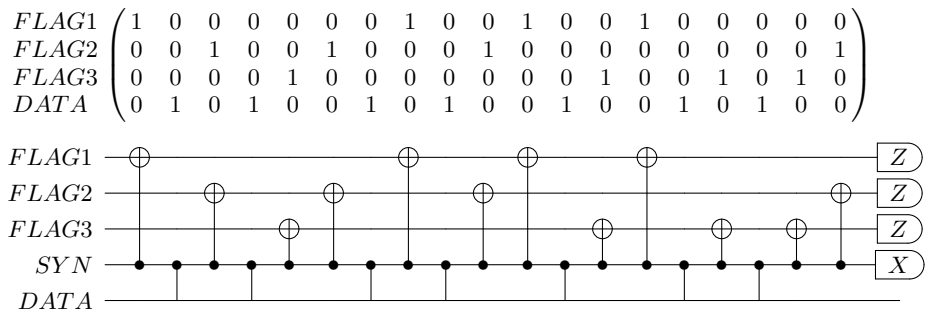


Given such a circuit matrix $C$ with dimensions $m \times n$, we can define the physical parity check matrix $F^{c}$. This matrix $F^{c}$ encodes how the syndrome errors propagate to the flag qubits, just as $F$ does, but corresponds to a concrete physically implementable circuit $C$. The matrix $F^{c}$ is $m - 1 \times n$ and has as element $(i, j)$ the sum $ \sum_{k = n}^j C_{i + 1, k}$ (over $\Z_2$). 

In this definition, we do not keep the row that encodes the location of the CNOTs connecting to the data. However, other than neglecting the data qubits, $F^{c}$ completely characterizes $C$, so we can work with the two representations interchangeably. In our construction, we start with a matrix $F$ based upon some classical code and modify it to produce $F^c$ so that $C$ (the circuit matrix) satisfies some set of physically motivated constraints (described in Section~\ref{sec:phys_const}).

Analogously, we define $D^c$ as the matrix that describes how errors on the syndrome qubit propagate to the data relative to a given circuit $C$.

The fault-tolerant requirement of the syndrome measurement circuit can be summarized (in a similar form to Eq.~\ref{eq:FT}) using the circuit diagram matrix $C$ by requiring
\begin{multline}\label{eq:FTcondition:circuit}
    \min (
    \abs{D^c e_s \oplus R(F^c e_s \oplus e_f)},\\
    w - \abs{D^c e_s \oplus R(F^c e_s \oplus e_f)}
    ) \leq k
\end{multline}
for $ k := t_f + t_s \leq t $, where we have included the fact that the syndrome we measure is a stabilizer. Of course in the case that $F^c = F$ and $D^c = D$ this reduces to Eq.~\ref{eq:FT}.

\section{Flag Gadgets based on Classical Codes}\label{sec:ideal}
We now show that by constructing the flag gadgets according to the parity checks of appropriate classical codes we can fault-tolerantly identify and correct syndrome faults in measuring any stabilizer. For simplicity, in this section, we assume that multiple CNOTs can be performed simultaneously and only suffer weight $1$ faults -- we will remove these assumptions in Section~\ref{sec:phys_const}.

\subsection{Flagged Syndrome Extraction for a Distance \texorpdfstring{$3$}{3} Code}\label{sec:ideal_t=1}
We start with an illustrative example, in which we show how to fault-tolerantly measure a weight $w + 1$ stabilizer of a distance $3$ code using $3 \ceil{\log_2 w}$ flag qubits and one syndrome qubit. Consider $F = \begin{pmatrix} H\\\hline H\\ \hline H\end{pmatrix}$ where $H$ is the parity check matrix for the Hamming code on $w$ (physical) bits. Imposing no physical constraints, we let $F^c = F$. The corresponding circuit is then uniquely identified by specifying one data CNOT per column of $F$. An example circuit for $w = 8$ is shown in Figure~\ref{fig:ideal_hamming}.

\begin{figure}
    \centering
    \includegraphics[scale=0.75]{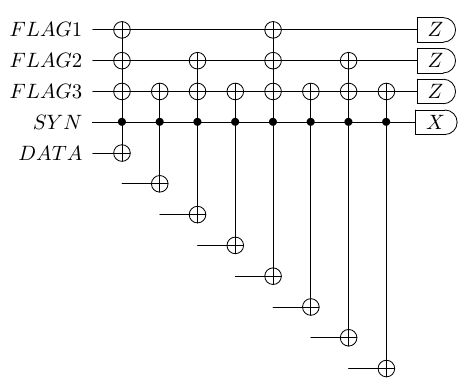}
    \caption{The parity check matrix for the $[7,4, 3]$ Hamming code implemented with CNOTs, with spaces between data CNOTs corresponding to physical bits of the codeword.}
    \label{fig:ideal_hamming}
\end{figure}

Clearly, if an error occurs on the syndrome qubit between data CNOTs $i$ and $i + 1$, the flag pattern produced is the same (replacing $1$s with $0$s and $-1$s with $1$s) as the result of the multiplication $ F e$, where $e$ is a binary vector with a single $1$ as the $i-$th entry. As such, we can fault-tolerantly identify the location of a single fault on the syndrome qubit, up to a stabilizer. The proof follows.
\begin{proof}
    Suppose the flag pattern is given by $P$.
    Consider two errors $e_1, e_2$ which produce the flag pattern $P$. We will show that the errors must propagate to the data identically (i.e. $De_1 = De_2$) up to a stabilizer. 

    For the purposes of this proof, call the trivial error (the no-error case) a syndrome error.
    Syndrome errors produce a flag pattern of weight either $0$ (corresponding to an error before or after all CNOTs composing the flag gadget, or the trivial error) or at least $3$, since columns of $F$ have at least $3$ nonzero entries (since columns of $H$ have at least $1$ nonzero entry).
    
    Consider first the case in which $e_1$ is a flag error, while $e_2$ is a syndrome error. A flag error produces a flag pattern with exactly one nonzero entry, so the flag patterns can not be the same and this case is impossible.

    Next, suppose $e_1$ and $e_2$ are both flag errors. Then they do not propagate to the data.

    Now suppose $e_1, e_2$ are both syndrome errors. If the flag pattern has zero nonzero entries, both errors propagate to a stabilizer. If the flag pattern is of weight at least $3$, both errors must be between the same pair of data CNOTs, since the columns of $F$ are distinct, and hence only an error between data CNOTs $i, i + 1$ produces column $F_i$. So $De_1 = De_2$.

    This shows that any two errors with the same flag pattern propagate to the same error on the data, so we can assign a unique correction to any flag pattern.
\end{proof}
\subsection{Flagged Syndrome Extraction for a distance \texorpdfstring{$d$}{d} code}\label{sec:ideal_gen_t}
The picture is similar when we consider more than $1$ error occurring during syndrome extraction. For a distance $d$ code, we use the flag gadget based on the binary matrix $F = \begin{pmatrix}
H\\
\hline
\ldots\\
\hline
H
\end{pmatrix}$ consisting of $d$ repetitions of $H$, where $H$ is a parity check matrix for the distance $d$ BCH code \cite{bc, hoc} (see Appendix~\ref{appendix:BCH} for a summary of the BCH code parity check matrix). Again we consider two errors $e_1$ and $e_2$ which produce the same (fixed) flag pattern $P$ and show that $De_1 = De_2$ up to a stabilizer.
\begin{proof}
    Since $H$ is the parity check matrix for a distance $d = 2t + 1$ code, $\sum_{k = 1}^{0 \leq n \leq 2t} H_k \not = 0$, where $H_k$ is the $k-$th column of $H$. This implies that $\abs{\sum_{k = 1}^{0 \leq n \leq 2t} F_k} \geq 2t + 1$. Consequently, any flag pattern $P$ has exactly $1$ associated set of columns $\{F_k\}_{0 \leq k \leq t}$, up to including the all-zeros column (corresponding to an error before or after all CNOTs), such that $\abs{P \oplus \sum_k F_k} \leq t$ (by the triangle inequality).
    Each column $F_k$ included implies that an odd number of syndrome errors have occurred between data CNOTs $k$ and $k + 1$. This implies that any two errors $e_1$ and $e_2$ with a different parity of syndrome errors between data CNOTs $i, i + 1$ for any $i$ (i.e. any two errors such that $De_1 \not = De_2$) have distinct flag patterns, and we can again assign a unique correction to each flag pattern.
\end{proof}
This shows that in the case where any number of CNOTs can act simultaneously, our construction is capable of perfectly identifying $De$ for any $e$ composed of syndrome and flag errors. In the next section, we show that in the absence of the ability to perform multiple CNOTs simultaneously, our construction is still able to approximately infer $De$, and that the approximation is fault-tolerant.
\section{Imposing Physical Constraints}\label{sec:phys_const}
When implementing a quantum error correcting code on hardware, there are many physical constraints to work around; for instance, connectivity constraints, constraints on how many qubits can be involved in a gate, and constraints on how many gates can act simultaneously. 

We consider a physically motivated constraint which we believe to be one of the more difficult ones to overcome: two CNOTs with a common control or target qubit cannot be applied simultaneously. We show that this constraint does not affect our construction. First, in Section~\ref{sec:unfolding} we provide a method to modify a given parity check matrix so that the corresponding circuit does not use more than one CNOT at once, then in Section~\ref{sec:general_gadgets} we show that this modification still allows for fault-tolerantly identifying syndrome faults by repeating a classical parity check matrix.

\subsection{Unfolding a parity check matrix}\label{sec:unfolding}
In this section, we describe an ``unfolding" procedure to construct $C$, or equivalently $F^{c}$, matrices which satisfy Eq.~\ref{eq:FTcondition:circuit} when starting with an $F$ matrix of a certain form, but which do not require performing multiple CNOTs simultaneously.
We will describe a form for $F$ such that this unfolding procedure is sufficient to satisfy Eq.~\ref{eq:FTcondition:circuit} in Section~\ref{sec:general_gadgets}. We assume $F$ is some number of parity check matrices for a BCH code \cite{hoc,bc} stacked on top of each other. See Appendix~\ref{appendix:BCH} for a brief summary of BCH codes.

Given an $F$ matrix, we first append a zero column to the right and left sides of the matrix. Then we add columns between adjacent columns of the appended matrix such that two consecutive columns differ in exactly one element. More precisely, we define a new matrix $F^c$ such that
\begin{align} 
\begin{cases}
F^c_{i,0}=F_i &\quad \forall i \\
F^c_{i,k}=F_i \oplus [F_i\oplus F_{i+1}][k] & \quad \forall i,k
\end{cases}
\end{align}	
Where the indices $i, k$ in $F^c_{i,k}$ together specify the column and $[F_i\oplus F_{i+1}][k]$ denotes the first $k$ non-zero elements of the column $F_i\oplus F_{i+1}$.
The matrix $F^c$ then has the form
\begin{widetext}
\begin{align}
F^c= \begin{pmatrix}
F_{1}, & F_{1}\oplus [F_{1}\oplus F_{2}][1], & \ldots \vert &  F_{2}, & F_{2}\oplus [F_{2}\oplus F_{3}][1], & \ldots \vert F_{m}, & F_{m}\oplus [F_{m}\oplus F_{m+1}][1], & \ldots \end{pmatrix}.
\end{align}
Each column of the unfolded parity check matrix is a stack of columns of $H$, with at most one subcolumn being neither a column of $H$ or the zero vector.
\end{widetext}

The matrix $F^c$ and the location of data CNOTS uniquely specify $C$ by just considering the difference of two consecutive columns of $F^c$. See Figure~\ref{fig:unfolding} for a small example.

\begin{figure*}[t!]
    \centering
    \begin{align*}
        F &= 
        \left(\begin{array}{c G G}
                   & 1 & 0 \\
            \cdots & 0 & 1  \\
                   & 1 & 1 \\
        \end{array}\right)
   \xmapsto{\text{appending } \vec 0}
   \left(
        \begin{array}{c G G B}
                   & 1 & 0 & 0 \\
            \cdots & 0 & 1 & 0  \\
                   & 1 & 1 & 0\\
        \end{array}\right)
        \left(
        \xmapsto{\text{unfolding}}
        \begin{array}{c G B G B B}
                   & 1 & 1 & 0 & 0 & 0 \\
            \cdots & 0 & 1 & 1 & 1 & 0  \\
                   & 1 & 1 & 1 & 0 & 0\\
        \end{array}\right) = F^c
\end{align*}
\begin{align*}
    C &= \begin{pmatrix}
              & 1& 0 & 0 & 1 & 0 & 0 & 1 &\\
        \cdots& 0& 0 & 1 & 0 & 0 & 0 & 0 &\\
        \cdots& 0& 1 & 0 & 0 & 0 & 1 & 0 & \\
              & 0& 0 & 0 & 0 & 1 & 0 & 0 &\\
\end{pmatrix}
    \end{align*}
    \caption{Example of unfolding procedure with columns of the original parity check matrix in green and transition columns in blue.}
    \label{fig:unfolding}
\end{figure*}

\subsection{General Fault Tolerant Flag Gadgets}\label{sec:general_gadgets}
We now use the unfolding method presented in Section~\ref{sec:unfolding} to construct general fault-tolerant flag gadgets. The proof is analogous to the proof in Section~\ref{sec:ideal_gen_t}, and follows from the observation that the independence of columns of $H$ guarantees that any two errors with the same flag pattern must be very similar.

Given an arbitrary $t$, define $F$ as the stack of $2t + 1$ parity check matrix $H$ of a distance $2t + 1$ error correcting code (for instance, the BCH code) on top of each other, i.e., define \[F = \begin{pmatrix}
    H_1 \\
    \hline
    \vdots\\
    \hline
    H_{2t + 1}
\end{pmatrix}.\]

As in Section~\ref{sec:ideal_t=1}, define $D$ to correspond to the physical implementation in which a data CNOT comes between each pair of columns of $F$.

Note that in this definition of $ D $, we do not distinguish between an error directly before a data CNOT and an error directly after a data CNOT. This is consistent with the fact that we can not distinguish these two errors by any flag construction and does not impact fault tolerance since misidentifying an error directly before a data CNOT as an error directly after a data CNOT (or vice versa) leads to at most one error on the data (which is allowable since the original number of errors is at least one).

\begin{lemma}\label{lem:same_flag_pattern}
Let $e_1, e_2$ be two errors that produce the same flag pattern $P$.
Denoting the vector of all syndrome errors by $e := e_{1, s} \oplus e_{2, s}$, we have $\abs{De_{2i}\oplus De_{2i + 1}}\leq 1$ for $e_j$ the $j-$th nonzero component of $e$.
\end{lemma}
\begin{proof} \label{improvement}
We will show that syndrome errors occur in pairs, with at most one data CNOT between members of the pair. 

 Write $P$ as $e_{1, f} \oplus \sum_{i \in I_1} F_{1, i}^c$ where $I_1$ is the set of indices corresponding to nonzero entries of $e_{1,s}$. This defines the set of columns $F_{1, i}^c$ corresponding to each syndrome error in $e_1$. Decomposing $P$ in terms of $e_2$ defines the set of columns $F_{2, i}^c$ corresponding to each syndrome error in $e_2$.

 Assuming that each $H_k$ has $r$ rows, define $F_{*, i}^k$ as the vector consisting of the $r(k - 1)$-th through $rk$-th entries of $F_{*, i}$, and similarly for $e_{*, f}^k$.
 
 Consider the set of all $H_k$ such that 
\begin{enumerate}
    \item \label{error_form_1}$F_{1, i}^k$ is a column of $H_k$, or is the zero vector,
    \item \label{error_form_2}$F_{2, i}^k$ is a column of $H_k$, or is the zero vector, and
    \item \label{error_form_3}$H_k$ is affected by zero flag errors (i.e. $e_{1, f}^k$ is the zero vector, as is $e_{2, f}^k$).
\end{enumerate}
    We first show that there exists at least one such $H_k$. 
    Start with the set $N = \{H_{k}\}_{k \leq 2t + 1}$. For each column $F_{1, i}$, at most one subcolumn of $F_{1, i}$ is both \emph{not} a column of $H$ and \emph{not} the zero vector.
    After removing all submatrices that do not satisfy \autoref{error_form_1}, we are left considering a set $N'$ of size at least $|N| - \abs{e_{1, s}}$, since $F_{*, i}^k$ is neither a column of $H_k$ or the zero vector for at most one $k$, and $i$ ranges from $1$ to $|e_{1, s}|$.
    Similarly, after removing all submatrices that do not satisfy \autoref{error_form_2}, we are left to consider at least $|N| - \abs{e_{1, s}} - \abs{e_{2, s}}$ submatrices.
    Finally, after removing all submatrices that do not satisfy \autoref{error_form_3}, we are left to consider at least $|N| - \abs{e_{1, s}} - \abs{e_{2, s}} - \abs{e_{1, f}} - \abs{e_{2, f}}$ submatrices. 
    
    Noting that $|N| = 2t + 1$, and that $\abs{e_{1, s}} + \abs{e_{2, s}} + \abs{e_{1, f}} + \abs{e_{2, f}} \leq 2t$, we see that we have at least $2t + 1 - 2t = 1$ submatrix satisfying \autoref{error_form_1}, \autoref{error_form_2} and \autoref{error_form_3}.
    
    We can now let $k$ be any index such that $H_k$ satisfies the three conditions given. 
    We use the columns of $H_k$ to characterize the relation between $e_{1, s}$ and $e_{2, s}$. 

    Recall that $F_{1, i}^k$ and $F_{2, i}^k$ are either columns of $H_k$ or the zero vector by assumption.
    Since $e_1$ and $e_2$ have the same syndrome and $e_{1, f}^k = e_{2, f}^k = 0$ by assumption, we must have $\bigoplus_{i\in I_1} F_{1, i}^k \oplus \bigoplus_{i \in I_2} F_{2, i}^k = 0$.
    Since $H$ has distance $2t + 1$, the sum of any \emph{distinct} $2t$ of its columns is nonzero. Since the sum \emph{is} zero, the columns must be non-distinct. In particular, either $F_{*, i}^k = 0$ or $F_{*, i}^k = F_{*, j}^k$ with a unique $j$ paired to each $i$. 
    By Appendix~\ref{appendix:ordering} we can assume that zero vector columns correspond to errors before or after measurement, which will propagate to a stabilizer. Each pair of errors corresponding to a pair of identical columns must have at most one data CNOT between them, since columns of $H$ are distinct, and each column is expanded to cover exactly one data CNOT. This proves the lemma.
\end{proof}
\begin{lemma} \label{lem:correction}
    Let $S$ be an arbitrary flag pattern.
    Consider the set of all errors that produce the given flag pattern, $E := \{ e = e_s + e_f : F^c e_s \oplus e_f = S, |e_s| + |e_f| \leq t\}$. Then for any $e_1, e_2 \in E$ such that $\abs{e_1} \leq \abs{e_2}$ we have $\abs{De_{1, s} \oplus De_{2, s}}\leq \abs{e_{2, s}} \leq \abs{e_{2}}$.
\end{lemma}
\begin{proof}
Let $e_{1}, e_{2} \in E$ have syndrome error weights $\abs{e_{1, s}} = k$,  $\abs{e_{2, s}} = n$ with $k \leq n$. 
 
We know that each component error (except for possibly one in the boundary, which does not affect $De$ anyway) has a paired error with at most one data CNOT in between by the previous lemma. Consider all of the component errors of $e_{1, s}$ which have their pair in $e_{2, s}$. If there are $l$ of these, $De_{1, s} \oplus De_{2, s}$ has weight at most $l + \frac{n - l}{2} + \frac{k - l}{2} = l - \frac{n + k - 2l}{2} = \frac{n + k}{2} \leq \frac{2n}{2} = n$. 

We get the term on the left-hand side by considering the error $e_{2, s}$ and applying the error $e_{1, s}$ as the correction. Then we note that:
\begin{enumerate}
    \item When we apply $e_{1, s}$ to the syndrome qubit the $l$ corrections that have their pair in $e_{2, s}$ each introduce at most one error on the data (corresponding to the data CNOT possibly in the middle of the pair).
    \item The $k - l$ errors that are not part of $e_{2, s}$ but which are included in the correction consist of pairs, and so their correction introduces at most $\frac{k - l}{2}$ errors.
    \item The ${n - l}$ errors which are part of $e_{2, s}$ which we do not correct are pairs (separated by at most one data CNOT), and hence introduce at most $\frac{n - l}2$ errors, which stay on the data after correction.
\end{enumerate}
This proves the lemma.
\end{proof}
\begin{corollary} \label{cor:main}
$F^c$ satisfies condition~\ref{eq:FTcondition:circuit}.
\end{corollary}
\begin{proof}
Again letting $S$ be arbitrary, consider the set of all errors that produce $S$, namely $E := \{ e = e_f + e_s : F^c e_s \oplus e_f = S, t_s + t_f \leq t\}$. Then let the correction operator be given by $D^c c$ where $c$ is defined as $\argmin_{e \in E} \abs{e}$. We wish to show that $\abs{D^c e_s \oplus R(F^ce_s \oplus e_f)} \leq t_s + t_f$ for any $e \in E$. 
Rewriting, we wish to show that $\abs{D^c c \oplus D^c e_s} \leq t_s + t_f$. This holds by Lemma~\ref{lem:correction} since by definition $\abs{c} \leq \abs{e}$.
\end{proof}

\subsection{Application to non-CSS syndrome measurements}\label{sec:nonCSS}

Suppose we wish to measure an operator $X^{\otimes n} Z^{\otimes m}$. Note that this applies to the measurement of $Y$ type operators if $n$ and $m$ are not disjoint. To measure this operator, usually some of the CNOTs (corresponding to the $X$ terms in the syndrome) are replaced with controlled phase gates, while the syndrome qubit is still measured in the $X$ basis. This means that the only type of error that can propagate from the syndrome qubit to the data is still $X$ type error. As such, we can use exactly the same set of flags that we use to protect the syndrome for a CSS code. 

In Figure~\ref{fig:5qubit}, this is illustrated for measuring the $XZZXI$ syndrome of the $\llbracket 5, 1, 3 \rrbracket$ perfect code.

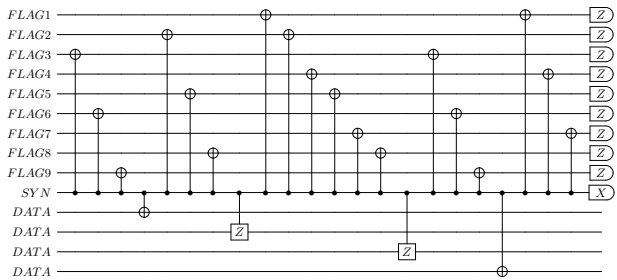
\begin{figure}[h]
    \centering
    \resizebox{0.45\textwidth}{!}{\input{figures/fivequbit_flag}}
    \caption{\label{fig:5qubit} An example of our flag construction for measuring the $XZZXI$ syndrome of the five qubit perfect code. Note that if we were to measure $XXXXI$ the flag circuit would be identical.}
\end{figure}

\subsection{Cost analysis to implement a single stabilizer}\label{sec:cost}
Having produced a fault-tolerant flag qubit construction, we analyze the circumstances in which it outperforms other constructions to protect a single stabilizer measurement. We mainly focus on the number of flag qubits used, as for near-term small-scale devices the number of qubits is severely limited, and also for large-scale devices reducing the number of qubits needed for fault-tolerance allows for more efficient use of hardware.

In our construction, the number of flag qubits used is $(2t + 1) n(w, t)$, where $n(w, t)$ is the number of rows in the classical parity check matrix. For the BCH code on $w$ bits with distance $2t + 1$, the number of rows is given by $t \log_2(w)$ (see Appendix~\ref{appendix:BCH}) for a total cost of $(2t^2 + t)\log_2(w)$ flag qubits.

For $t = O(1)$, this is asymptotically logarithmic in the weight of the stabilizer, which is optimal. However, for $O(1) = w \approx t$, the constant factors can be quite large. For example, for $w = 15, t = 2$, the number of flag qubits used is $(10)(4) = 40$, which is outperformed by the various schemes for fault-tolerance which scale linearly in the weight of the syndrome. Even reducing the number of repetitions to $t + 2$, as in Section~\ref{sec:sims}, yields $32$ flag qubits. It is only once $w \geq 70$ (or $w \geq 48$ assuming that we only need to repeat $t + 2$ times) that our construction outperforms schemes that scale linearly in the weight of the syndrome.

This suggests that exponential saving provided by our constructions only shows up in codes with high-weight stabilizers, such as Bacon-Shor \cite{baconshor} codes on certain choices of the lattice. However, in Section~\ref{sec:multi_syn} we present a slight modification to the syndrome extraction procedure that allows our construction to protect multiple syndromes at once. Under some mild assumptions about the speed at which qubits can be reset, this modification yields a large reduction in qubits needed even for small codes such as the $\llbracket 5, 1, 3 \rrbracket$ code.

We now briefly address the cost in CNOTs to implement our construction. As a very pessimistic upper bound, each row of the parity check matrix $H$ will alternate between $0$ and $1$, requiring $w$ CNOTs per row for a weight $w$ syndrome. With $(2t + 1) t \ceil{\log_2 w}$ total rows, this implies a maximum CNOT cost of $(2t + 1)t w \ceil{\log_2 w} \sim t^2 w \log_2 w$. Realistic parity check matrices do not have rows that alternate between $0$ and $1$, making this upper bound very loose. 

\section{Decoding Procedures}\label{sec:decoding}
Given a certain flag pattern, we must be able to infer a correction to apply to the data; that is, we must decode the flag pattern. In this section, we outline two different decoding algorithms.

\subsection{Brute Force}\label{sec:brute_force}
Section~\ref{sec:general_gadgets} implicitly defines a decoder.
We construct a table of errors and how they propagate. We can consider the set of all errors $e_1, \ldots, e_n$ that have the same syndrome (i.e. $F^c e_i = F^c e_j$ for all $i, j$). Since we can not distinguish between $e_1, \ldots, e_n$ we need to find some correction that will correct any $e_i$ fault-tolerantly.

We can phrase this condition in terms of Hamming balls drawn about the data error that each $e_i$ propagates to. Succinctly, there exists a fault-tolerant recovery map $R$ if and only if, for arbitrary flag pattern $P$ we have 
\[\bigcap_{e_i \in  \{e : F^ce = P\}} B_{\abs{e_i}}(D^c e_i) \not = \emptyset\] where $B_r(d)$ is the Hamming ball with radius $r$ around the point $d$ in the space of all possible errors (or corrections) on the data. This is obvious if we observe that $B_{\abs{e_i}}(D^c e_i)$ is the set of corrections that do not increase the number of errors on the data from the original number of physical errors. This condition is illustrated in Figure~\ref{fig:decoder}.
\begin{figure}[t]
    \centering
    \includegraphics[scale=0.8]{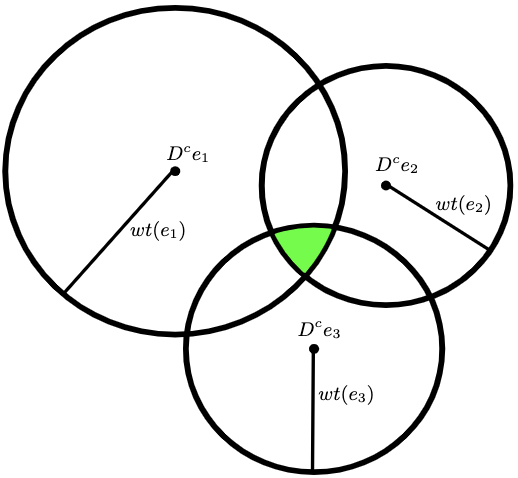}
    \caption{\label{fig:decoder}An example of the decoding process for a set of errors $e_1, e_2, e_3$ with a common syndrome. In green is the set of all corrections which do not increase the number of errors on the data regardless of which physical error produced the observed syndrome. Note that using at least $2t + 1$ repetitions with our construction, a picture like this will not arise. Instead, the center of each small ball will be contained in each larger ball.}
\end{figure}

The brute force decoding algorithm simply computes the intersection of all Hamming balls with appropriate radii centered at the data errors arising from collections of faults with the same flag pattern and picks an arbitrary member of the intersection as the correction to apply.

\subsection{Majority Vote Decoder}\label{sec:maj_vote}
Another natural procedure presents itself. Intuitively, we can use the repetition of the parity check matrix $H$ to correct the errors on the flag pattern, then use the new flag pattern obtained to apply a correction that does not increase the number of errors on the data from the original number of physical errors based by decoding the classical code hidden inside our flag qubits.

\begin{lemma}
    If $F$ is composed of $(t + 1)^2$ repetitions of a classical parity check matrix $H$, the following correction procedure produces a fault-tolerant correction for any error-producing flag pattern $P$:
    \begin{enumerate}
        \item Divide $P$ into $(t + 1)^2$ subcolumns, $P_i, 1 \leq i \leq (t + 1)^2$
        \item Let $P^*$ be the subcolumn which occurs most frequently
        \item Taking $P^*$ as the syndrome for the classical parity check matrix $H$, decode $P^*$ and apply the correction indicated.
    \end{enumerate}
\end{lemma}
\begin{proof}
Note that for any error $e$ the classically computed correction associated with a subpattern produced by a repetition where $e$ has support entirely on columns of $H$ or zero columns (an ``error-free'' repetition) is a fault-tolerant correction for $e$. 
This follows from the fact that in an error-free repetition, the correction produced by the classical decoder exactly identifies the parity of syndrome errors in the area covered by each column of $H$. So at worst the correction produced neglects to correct some even number of errors in each region covered by a single column of $H$. As before, the area covered covers at most one data CNOT, and hence neglecting to correct these errors still yields a fault-tolerant correction. 
So it is enough to show that the most common subpattern is produced by an error-free repetition.

If $e$ is incident upon a column of repetition $i$ which is not a column of $H$, and $e$ is incident upon a column of repetition $j$ which is not a column of $H$, but for all repetitions $k$ for $i < k < j$ all errors are incident upon a column of $H$ or the zero column and no repetition $k$ suffers a flag error, all subpatterns from repetitions $i +1$ to $j - 1$ are identical. So a single error $e_i \in e$ can divide a region of repetitions into two regions with possibly different patterns, with some arbitrary subpattern in between. So it is enough to show that dividing a region of length $(t + 1)^2 - t$ into $t$ subregions produces at least one subregion of length $t + 1$. Since $\frac{( t + 1)^2 - t}{t} > t + 2 - 1 = t+ 1$ this condition is satisfied. So the most common subpattern has multiplicity at least $t + 1$, and hence can only be produced by $ t + 1 $ error-free repetitions in a row.
\end{proof}

\subsection{Intermediate Decoder}\label{sec:intermediate}
Finally, we present an intermediate decoding algorithm. This algorithm still scales exponentially with $t$, but not as badly as the brute force version. It also only requires $2t + 1$ repetitions of the parity check matrix, as the brute force version does. 

Again we call a repetition ``error-free'' if it suffers no flag errors, and each syndrome error falls on one of the columns of the parity check matrix (not in between).
We first decode each subpattern with a decoder for the BCH code (skipping any which can not be decoded). 
Write the decoded error as a bit string of syndrome errors. 
The results of decoding any two error-free repetitions must differ by a bit string of the form $\sum_{i = 1}^{k <= n} 0^{f_1(i)}11 0^{f_2(i)}$ where $n$ is the number of non-error-free repetitions between them and $0^{f_{0, 1}(i)}$ denotes some number of zeros which depends on $i$.

The locations of the $11$ substrings in each term of the sum in the difference vector correspond to the location of a transition error. 
If we can find a set of $t + 1$ subpatterns which all differ from each other in the appropriate way, then at least one of them must be error-free. 
Recall that a minimum weight error which is error-free in repetition $i$ and produces the flag pattern $P$ is a fault-tolerant correction for any other error which is error-free in repetition $i$ and produces $P$. 
So the problem is reduced to finding a minimal number of $0^{f_1(i)}11 0^{f_2(i)}$ strings and assigning them each to a repetition, such that the remaining repetitions differ from each according to the sum of difference vectors between them. After finding such a set of difference vectors, we can choose any of the remaining repetitions as the correction.

\subsection{Comparison of Decoding Costs}
The brute force decoder described in Section~\ref{sec:brute_force} requires compiling a table of all errors and their syndromes. Clearly, the dominant term in the time cost for this decoder is the cost of considering every error. The number of distinct errors is given by 
\[ \sum_{k = 0}^t \binom{n}{k} \]
where $ n $ is the total number of CNOTs used in the circuit (equivalently, the weight of $ C $), which is proportional to the number of flag qubits $ f $. This scales exponentially with the number of flag qubits. If $l$ is the least number of CNOTs attached to a flag qubit, then this decoding algorithm requires $\Omega(\sum_{k = 0}^t \binom{l(2t + 1) + w}{k})$ operations. Using the bound $\sum_{k = 0}^t \binom{n}{k} = \Omega(2^{H(t/n)n})$ we can lower bound this by $\Omega(2^{l(2t)}2^{ H(l)}) = \Omega(2^{2lt})$ where $H(x)$ is the binary entropy of $0 \leq x \leq 1$.

However, the majority vote decoding procedure in Section~\ref{sec:maj_vote} only requires finding the most frequently occurring subpattern of length $t\log w$, before decoding one classical syndrome.

In the case of the BCH code, the decoding algorithm requires $ O(wt) $ operations \cite{BCHdecoder}. Finding the most frequent subpattern takes $O(t^5 \log(w))$ operations for a total complexity of $O(t^5 \log w + w t) = O(t^5 \log w)$.

Finally the intermediate decoder in Section~\ref{sec:intermediate} has complexity $O( \binom{w}{t} \binom{2t + 1}{t})$, where the first term is the number of ways to choose the left/right location of the transition columns, and the second term is the number of ways to choose the repetition each transition column falls in.

\section{Protecting Several Stabilizer Measurements}\label{sec:multi_syn}
In general, several syndromes need to be extracted consecutively, and the entire sequence needs to be repeated to account for measurement errors (in the Shor error correction picture). In this section, we show that the flag qubit construction in Section~\ref{sec:general_gadgets} can be trivially modified to identify error locations in this whole sequence of repeated measurements. Because our construction enjoys sublinear scaling, flagging the entire sequence of syndrome extraction requires fewer qubits than flagging each measurement separately.

\subsection{Connecting and Flagging a Sequence of Stabilizers}
Consider a syndrome extraction circuit in the form of Figure~\ref{fig:syndrome_ex}. We can modify this circuit by adding CNOTs connecting subsequent pairs of syndrome qubits, without changing the measurement outcomes observed as in Figure~\ref{fig:connected_syndrome_ex}.
This modification means that $X$ errors on syndrome qubit $i$ will propagate to all syndrome qubits $j > i$. However, the measurement outcomes of these syndrome qubits will be unaffected since no $Z$ errors propagate and the syndrome qubit is measured in the $X$ basis. Additionally, $X$ errors that propagate from one syndrome qubit to another propagate to trivial errors (stabilizers) on the data. Adding this modification, however, means that a flag gadget starting on syndrome qubit $i$ and ending on syndrome qubit $i > j$, as in Figure~\ref{fig:connected_syndrome_flag} will flag if there are an odd number of errors in the range it covers, regardless of which syndrome qubit they occur on (where we interpret error occurring after the connecting CNOT as unflagged measurement errors). 

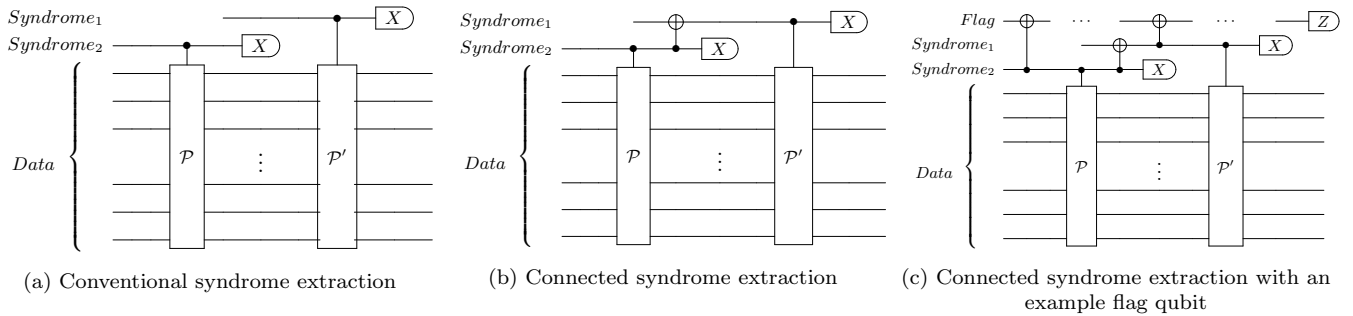
\begin{figure*}
    \centering
    \subfloat[][\label{fig:syndrome_ex}Conventional syndrome extraction]{
        \resizebox{0.33\textwidth}{!}{
        \import{figures/}{syndrome_extraction}}}
    \subfloat[][\label{fig:connected_syndrome_ex}Connected syndrome extraction]{
        \resizebox{0.33\textwidth}{!}{
        \import{figures/}{connected_syndrome_extraction}}}
    \subfloat[][\label{fig:connected_syndrome_flag}Connected syndrome extraction with an example flag qubit]{
        \resizebox{0.33\textwidth}{!}{
        \import{figures/}{connected_syndrome_extraction_flagged}}}
    \caption{Syndrome extraction of $\mathcal{P}$ and $\mathcal{P'}$ (where $\mathcal{P}, \mathcal{P'}$ are multiqubit Pauli operators)}
\end{figure*}

Suppose that we wish to measure $L$ stabilizers, each with weight $w_i$. 
To ensure fault-tolerance, each round of measurement must be repeated $s$ times, where $s$ depends on the properties of the code. Usually $s$ is taken to be $(t + 1)^2$, so that if at most $t$ errors occur we are guaranteed to see the syndrome corresponding to a fault-tolerant correction appearing $t + 1$ times in a row \cite{bys, shor}.
We define the total weight $W$ as $\sum_{i = 1}^L w_i$. 
Connecting all of these extractions together in the form outlined above leads to a circuit with $sW$ distinct locations that need to be identified by the flag gadgets. Our construction in Section~\ref{sec:general_gadgets} holds, and we can identify up to $t$ errors distributed between these $sW$ locations using only $(2t + 1)t\log(sW)$ flag qubits.

\subsection{Error Propagation and Time Ordering}

Since we only measure flag qubits after all rounds of syndrome extraction, we allow low-weight errors on the syndrome bits to propagate to high-weight errors on the data. This means that some of the syndrome bits measured will be incorrect since we do not measure the flags until all syndromes have been extracted. However, since the syndrome extraction circuit is a Clifford circuit, after measuring the flags we can track the propagation of errors classically and correct any syndrome bits necessary before applying the correction implied by the (updated) syndromes. 
Additionally, since we do not measure flags until after all syndromes have been extracted, we can not stop early if we see $t + 1 $ identical syndromes in a row - instead, we have to continue measuring syndromes for a full $s$ repetitions.
We treat the fault-tolerance of the connected syndrome picture more thoroughly in Appendix~\ref{appendix:connected}. 

This method of connecting syndrome qubits naively increases the time required for a round of syndrome extraction, since the connections induce a time ordering which disallows extracting syndrome bits in parallel. However, the time ordering is not necessary with the introduction of an extra ancilla qubit per syndrome qubit. Instead of connecting two syndrome qubits together, we can instead connect a syndrome qubit to an ancilla in the same way, measure the ancilla in the $Z$ basis, then adjust the flag measurements according to the measurement result (simulating the propagation of the possible $X$ error through the CNOTs classically). This removes any time ordering, and allows for a flag circuit based on our construction to protect multiple syndrome extractions acting in parallel.

\section{Resource Analysis for Connected Syndromes}\label{sec:multi_syn_resource}
Consider a code defined by $L$ (independent) stabilizer generators, each with weight $w$. Assuming that qubits can not be reset (see Section~\ref{sec:reset_time} for a discussion on when this constraint is relaxed), the total number of qubits required to do $s$ rounds of syndrome extraction using Shor error correction or similar methods is linear in $Lsw$. In contrast, the number of qubits required by our construction for the fault tolerant extraction is given by $O(t^2 \log(Lsw))$, where $2t+1$ is the distance of the code. As discussed in Section \ref{sec:cost}, the construction clearly provides an advantage when the weight of the stabilizer $w$ is large. In this section, we show that by connecting multiple syndromes, similar advantage is gained for codes with small weight $w$, such a qLDPC codes.

It is important to emphasize that one main advantage of our construction is that we only need to prepare simple single-qubit states, in contrast to Shor, Steane, or Knill error correction, in which access to highly entangled fault-tolerantly prepared ancilla states is assumed. In particular, in these schemes, the required resources to produce such an entangled state can significantly increase with the distance of the code $2t+1$.  Throughout all comparisons in this section, we have neglected the cost of preparing the ancilla states. Including the cost of fault-tolerantly preparing the ancilla states would increase the resource savings offered by our construction.

\subsection{Applicability to qLDPC codes}
Any family of codes with parameters that result in a small ratio $\frac{t^2 \log Lsw}{Lsw}$ can benefit from our construction. One important family to consider is qLDPC codes \cite{breuckmann2021ldpc} where the weight of each stabilizer is a constant independent of the number of qubits $n$, i.e. $w=O(1)$. Note that this constraint also fixes the number of independent stabilizers to be  $L=\Theta(n)$, since each qubit must be included in at least one stabilizer for the code to be able to correct any single error. For this family of codes, we make the following observation.

\textbf{Remark:} For any qLDPC code, our construction requires exponentially fewer ancilla qubits compared to the Shor style syndrome extraction.

To see this, note that for Shor style syndrome extraction we have $s=\Theta(t^2)$ and therefore the ratio scales as $\Theta (\frac{\log L t^2}{L})$ which is bounded by $\Theta (\frac{\log n}{n})$ for any error-correcting qLDPC code.


The improvement persists when the overhead is compared to schemes other than Shor style syndrome extraction, that use less than $s=\Theta(t^2)$ measurements for each stabilizer. For example, consider hypergraph-product codes \cite{tillich2014HGP} whose distance scales as $t = \Theta(\sqrt n)$. For this family of codes, even with $s = \Theta(t)$ repetitions (inspired by the number of measurements required for the surface code), the ratio becomes $\Theta(\frac{\log n}{n^{1/2}})$, which still shows an exponential improvement.

A regime in which our scheme does not provide an advantage for qLDPC codes is when they satisfy the single-shot property ($s=O(1)$) and have a linear distance ($t = \Theta(n)$).

\subsection{Small-size codes}

Now that we have seen the asymptotic advantage of the scheme in a wide range of code parameters, in this section we investigate the advantage for small-size codes. We present the number of qubits required for some small codes of interest in Table~\ref{table:requirements}. (Note that \cite{1} provides a Shor-style method to do fault-tolerant syndrome extraction of a syndrome for a distance 3 code using $\max\left(3, \ceil{ \frac w 2}\right)$ qubits. This is not considered in Tables~\ref{table:requirements} and~\ref{table:requirements_nice_decoder}.)
\begin{table}
\begin{tabular}{c| c  c c| c |c}
                                           &     &     &     &  Shor EC   & Flag EC\\          
                                           & $t$ & $s$ & $W$ & $sW$& $(2t + 1) t \log(sW)$\\
                                           \hline
     $ \llbracket 5, 1, 3 \rrbracket $ code& 1   & $4$ & $16$& $64$ & $18$ \\
     $ \llbracket 7, 1, 3 \rrbracket $ code& 1   & $4$ & $24$& $96$ & $21$ \\
     $ \llbracket 19, 1, 5 \rrbracket$ code& 2   & $9$ & $84$& $756$ & $100$ \\
\end{tabular}
\caption{\label{table:requirements} Comparison of the number of extra qubits needed to do fault-tolerant syndrome extraction with and without flag qubits}
\end{table}

For large $sW$, the brute force decoder outlined in Section~\ref{sec:sims} is computationally intensive. As such, it might be desirable to instead use $(t + 1)^2$ repetitions of the classical parity check matrix $H$ to construct $F$, so that the more efficient decoder presented in Section~\ref{sec:decoding} may be used. This increases the number of qubits required to $(t + 1)^2 t \log(sW)$, but for large $sW$ this increased number is still less than the requirements for Shor error correction, or similar methods. The number of qubits required to use this decoder is compared to Shor error correction in Table~\ref{table:requirements_nice_decoder}.
\begin{table}
\begin{tabular}{c| c  c c| c |c}
                                           &     &     &     &  Shor EC   & Flag EC\\          
                                           & $t$ & $s$ & $W$ & $sW$& $(t+1)^2 t \log(sW)$\\
                                           \hline
     $ \llbracket 5, 1, 3 \rrbracket $ code& 1   & $4$ & $16$& $64$ & $24$ \\
     $ \llbracket 7, 1, 3 \rrbracket $ code& 1   & $4$ & $24$& $96$ & $28$ \\
     $ \llbracket 19, 1, 5 \rrbracket$ code& 2   & $9$ & $84$& $756$ & $180$ \\
\end{tabular}
\caption{\label{table:requirements_nice_decoder} Comparison of number of extra qubits needed to do fault-tolerant syndrome extraction with and without flag qubits using the decoder specified in Section~\ref{sec:decoding}.}
\end{table}

For distance three codes, there is a flag construction that uses an asymptotically logarithmic number of qubits with a smaller constant factor than our construction \cite{4}.
Using the same method we present to connect syndrome qubits together, the construction in \cite{4} can be leveraged to use even fewer qubits than our method for distance $3$ codes.

Finally, let us compare this method of fault-tolerant syndrome extraction to some non-Shor-style error correction. Specifically, consider extracting syndromes fault-tolerantly using either Knill \cite{knill} or Steane \cite{steane} error correction. Steane EC requires a block of $n$ qubits to extract a single syndrome (where $n$ is the number of physical qubits in the code), while Knill EC requires two blocks. However, neither method requires repeated syndrome extraction. In Table~\ref{table:req_knill_steane} we provide the number of extra qubits required to do a round of error correction using these methods (neglecting ancilla qubits used to prepare the block(s) of qubits, and assuming no reset is possible). We see that our method compares favorably with these methods as well.

\begin{table}[ht]
\begin{tabular}{c | c  c| c c}
                                                 
                                           & n & $ \text{number of syndromes}$ & Steane EC & Knill EC\\
                                           \hline
     $ \llbracket 5, 1, 3 \rrbracket $ code& 5& 4 & 20 & 40\\
     $ \llbracket 7, 1, 3 \rrbracket $ code& 7& 6 & 42 & 84\\
     $ \llbracket 19, 1, 5 \rrbracket$ code& 19& 18 & 342 & 684\\
\end{tabular}
\caption{\label{table:req_knill_steane} Number of qubits required to do a round of error correction using the Steane and Knill methods.}
\end{table}

It also is important to note that it is not always necessary to measure all of the stabilizers $( t + 1)^2$ times in Shor-style error correction \cite{bys, encsyn, ss}. Given that our construction is agnostic to the form of the measured stabilizers, we can directly apply it to an error correction procedure which uses fewer repetitions, such as a recent scheme for short syndrome measurement sequences~\cite{adaptive_syn}. In our resource estimates in this section, we have considered using $(t + 1)^2$ rounds of syndrome extraction.

\section{Computer Simulations}\label{sec:sims}
With the framework provided by Section~\ref{sec:FT_req}, we can do computer-aided searches for fault-tolerant flag gadget constructions. In this section, we search for stabilizer weights where it is possible to use fewer than $2t + 1$ repetitions of the BCH code parity check matrix as well as very small examples for distance $5$ codes.

The construction given in Section~\ref{sec:general_gadgets} only gives an upper bound on the number of repetitions required; it does not always take $2t + 1$ repetitions in order to achieve fault-tolerance for up to $t$ errors. 
 In particular, for measuring a weight $15$ syndrome, fault-tolerance is achieved using only $t + 1$ or $t + 2$ repetitions of the BCH code, as opposed to $2t + 1$ (see Table~\ref{table:simresults}).
\begin{table}
    \centering
    \begin{tabular}{c|cccccc}
     t\textbackslash r& 2 & 3 & 4 & 5\\
     \hline
                     1& Y & Y & Y & Y\\
                     2& N & Y & Y & Y\\
                     3& N & N & N & Y\\
\end{tabular}
    \caption{Fault-tolerance results for $w = 15$, $d = 2t + 1$ and $r$ repetitions of the parity check matrix for a distance $d$ BCH code. `N' indicates non-fault-tolerant parameters, while `Y' indicates fault-tolerant parameters.}
    \label{table:simresults}
\end{table}
This shows that certain parity check matrices have properties that reduce the negative effects of columns produced by transitioning from one column of the check matrix to another.

Similarly, numerical simulations show that 3 repetitions of the $2$ error correcting BCH code on 13 bits are sufficient to deduce fault-tolerant corrections for extracting a weight $13$ stabilizer of a distance $5$ code. In Appendix~\ref{appendix:double} we argue that we can double the weight of the stabilizer by making a small modification to the correction procedure, but no change to the flag qubit pattern. This is not guaranteed to hold when using $3$ repetitions of the parity check matrix instead of $2t + 1 = 5$, but simulations show that in this case, replacing the weight 13 stabilizer with a weight 26 stabilizer and using 3 repetitions of the BCH code on $13$ bits does indeed allow for fault-tolerant corrections. This is significant since the parity check matrix for the BCH code on $13$ bits has $8$ rows. This means we use $24$ flag qubits and one syndrome qubit for a total cost of $25$ qubits to extract a weight $26$ stabilizer, which is less costly than Shor error correction (which would use $26$ qubits). Since in this example $t = 2$, previous flag qubit constructions do not apply (assuming slow reset).

Numerical simulations also show that it is possible to fault-tolerantly extract an arbitrary weight $5$ stabilizer for a distance $5$ code using only $2$ flag qubits, for a total cost of $3$ ancilla qubits. The circuit to do so is shown in Figure~\ref{fig:5qubit2flag} while the correction rules are shown in Table~\ref{table:5qubitrules}. The circuit and correction rules were both obtained by brute force. This method yields a $2$ qubit advantage over Shor syndrome extraction, and a $3$ qubit advantage over the alternative flag gadget construction given for stabilizer codes of any distance \cite{3}.

These examples are summarized in Table~\ref{table:small_examples}.
\begin{table}[h]
    \centering
    \begin{tabular}{c|ccc}
      w& 5 & 15 & 26 \\
      \hline
    qubits& 3 & 25 & 25\\
    repetitions& None & 3 & 3
\end{tabular}
    \caption{Total number of ancilla qubits (including the syndrome qubit) sufficient to measure a weight $w$ stabilizer of a distance $5$ code fault-tolerantly, and the number of repetitions of the BCH parity check matrix used. Note that the example for a weight $5$ stabilizer does not use the parity check matrix construction.}
    \label{table:small_examples}
\end{table}

\begin{figure}
    \centering
    \resizebox{0.4\textwidth}{!}{\input{figures/5qubit2flag}}
    \caption{Flag gadget for fault-tolerantly extracting a weight $5$ stabilizer.}
    \label{fig:5qubit2flag}
\end{figure}
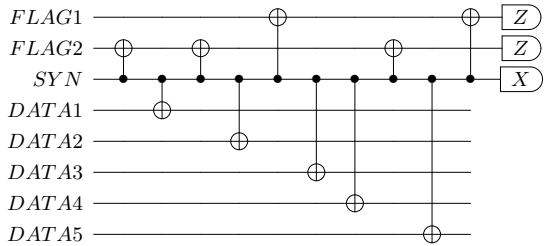

\begin{table}
    \centering
    \begin{tabular}{c|c}
        \thead{$FLAG1, FLAG2$ \\Measurements} & \thead{Correction\\ on $DATA$} \\
        \hline
         $+1, +1$&  $I$\\
         $+1, -1$&  $X_1$\\
         $-1, +1$&  $I$\\
         $-1, -1$&  $X_1X_2X_3$
    \end{tabular}
    \caption{Fault-tolerant correction rules for the circuit presented in Figure~\ref{fig:5qubit2flag}. Note that there are other choices for fault-tolerant correction rules.}
    \label{table:5qubitrules}
\end{table}

In general, this method of syndrome extraction yields performance improvements for any code with relatively high-weight syndromes (or high total weight, as in Section~\ref{sec:multi_syn}). One example of such a family of codes is produced 
when choosing the set of stabilizers measured according to some classical code \cite{ashikhmin, bys}. 
Generally, the weight of the stabilizers measured is much greater than the weight of the stabilizer generators -- for example, a distance $d$ (for $d$ sufficiently large) rotated surface code equipped with the $[16, *, d]$ BCH code measures stabilizers of weight approximately $28$. Other examples can be obtained from code concatenation, which leads to high-weight stabilizers being measured after only a few concatenations.
\section{Qubit Reset}\label{sec:reset_time}
Although we have focused on the regime in which qubit reset is slow or impractical, our construction can also provide a reduction in the number of qubits necessary for fault-tolerant error correction when qubit reset is practical. In general, as qubit measurement and reset gets faster, our construction provides fewer reductions.

Suppose that qubits can be reset in parallel in time $\tau$, and single qubit measurements can be performed in time $\mu$ (where a CNOT takes unit time). Note that if qubits can not be reset (in parallel) the method of connected syndrome extraction presented in Section~\ref{sec:multi_syn} is more useful than if qubits can be reset (in parallel).
 In Tables~\ref{table:requirements} and~\ref{table:requirements_nice_decoder} we implicitly assume that $\tau$ is arbitrarily large. 
 However, it is of course possible to reset qubits. We perform resource estimation for a round of error correction for Shor syndrome extraction and  flagged (connected) syndrome extraction. The maximum number of qubits required for flagged syndrome extraction is given simply by $(2t + 1)t \ceil{\log_2 (sW)}$ as in Section~\ref{sec:multi_syn}. To estimate the resources required for Shor syndrome extraction we use Algorithm~\ref{alg:resource_est}, which provides a lower bound on the number of qubits necessary when doing a round of error correction as fast as possible.

\begin{algorithm}[H]
\begin{algorithmic}
\State \INPUT{a set of generators $G$ and a stack of available qubits $p$}
\State used $\gets \emptyset$
\State busy $\gets \emptyset$
\For{each generator $g \in G$}
    \State move any qubits which have become non-busy from busy to top of $p$
    \State selected $\gets$ first $\abs{g}$ qubits from $p$
    \State used $\gets$ used $\cup$ selected
    \State busy $\gets$ selected and times selected qubits will become available again
\EndFor
\State \Return $|\text{used}|$
\end{algorithmic}
\caption[]{\label{alg:resource_est} Outline of our strategy for estimating resources required for Shor EC}
\end{algorithm}
 
To provide a concrete example, we consider the recent experiments on superconducting hardware \cite{google}, where measurement and reset times are approximately $500$ ns and $160$ ns respectively. With a CNOT time of approximately $13$ ns \cite{CNOT_time}, this corresponds to $\tau = 38.5, \mu = 12.3$. Resource estimation shows that for these values of $\tau, \mu$, our construction yields a qubit advantage over Shor EC for the $\llbracket 9, 1, 3\rrbracket$ Shor code, but not for the $\llbracket 7, 1, 3 \rrbracket$ Color code. This aligns with the fact that the stabilizers that need to be protected in the Shor code (the weight $6$ $X$-type stabilizers) are higher weight than the stabilizers for the Color code (which are weight $4$), meaning that Shor syndrome extraction requires fewer ancilla qubits to be used simultaneously when extracting the syndromes for the Color code. This again shows that our construction is especially useful when measuring high-weight stabilizers.
 
Other recent experiments on ion trap hardware \cite{monroe} give measurement times corresponding to $\mu = 13$. Although reset times are not given, we assume a comparable time scale and set $\tau = 13$ for resource estimation. In this regime, where CNOTs are relatively slow, we begin to see qubit advantages for our construction on a distance $3$ code when stabilizers are weight ${\sim}13$. 

It may be possible to achieve qubit advantages using our construction even in the case that $\tau$ and $\mu$ are relatively small and when measuring low-weight stabilizers by using fewer than $2t + 1$ repetitions of the parity check matrix. The numerical results in Section~\ref{sec:sims} show that it is not always necessary to use $2t + 1$ repetitions of the parity check matrix, and it is often sufficient to use $t + 1$ repetitions. 

\section{Conclusion}
In this work, we have developed a framework to design flag gadgets based on classical codes. This framework allows for enough freedom to achieve logarithmic cost scaling with the appropriate choice of code (e.g. the BCH code), to perform fault-tolerant syndrome measurement of any general quantum code.

To maximize the gain from this exponential
saving, we have proposed methods to fault-tolerantly measure multiple stabilizers using a single gadget. We proposed several small examples of the constructions using a computer-assisted search that can be appropriate for near-term experiments on small quantum computers. We discussed the overhead when slow qubit-reset is available and proposed several decoding strategies.



This work leaves several questions open. In Section~\ref{sec:general_gadgets} we assumed that columns introduced by the unfolding procedure in Section~\ref{sec:unfolding} were completely arbitrary, and hence used $2t + 1$ repetitions of the parity check matrix to guarantee that they would not cause non-fault-tolerant corrections. However, it may be possible to characterize these columns in a way that allows for fewer repetitions (or, in the extreme case, to remove these columns altogether). Incorporating other physically relevant constraints, such as geometric locality of interactions is another direction to pursue.

In addition,  our construction uses a two-step procedure, in which first the propagated errors from the syndrome onto data are corrected, then the syndrome bits are used to do error correction according to the quantum code. However, one can integrate the information from the flag qubits and the syndrome bits in order to directly infer the correction to apply to the data, possibly using fewer flag qubits than the two-step version, and can add structure between multiple syndrome qubits (as in \cite{steane_flag} for instance).

\section{acknowledgement}
This material is based upon work supported by the U.S. Department of Energy, Office of Science, National Quantum Information Science Research Centers, Quantum Systems Accelerator. Additional support by NSF CAREER award No. CCF-2237356 and NSF Grant No. CCF-1954960 is acknowledged.

\bibliographystyle{apsrev}
\bibliography{refs.bib}

\appendix
\section{BCH Codes}\label{appendix:BCH}
Given $m$ such that $2^m \geq w$, the BCH code correcting up to $t \leq \frac w 2$ errors out of $w$ locations uses at most $mt$ parity checks. The parity check matrix is constructed using a \emph{primitive element} $\alpha$ of $GF(2^m)$ (i.e. an element $\alpha$ of the finite field with $2^m$ elements such that every nonzero element of $GF(2^m)$ can be written as $\alpha^i$ for some $i$).

The (redundant) parity check matrix $H'$ is defined as
\[
    H' := \begin{pmatrix}
    1 & \alpha      & \alpha^2        & \ldots & \alpha^{w - 1}\\
    1 & \alpha^2    & (\alpha^2)^2    & \ldots & (\alpha^3)^{w - 1}\\
      &             &                 &\vdots\\
    1 & \alpha^{2t} & (\alpha^{2t})^2 & \ldots & (\alpha^{2t})^{w - 1}
    \end{pmatrix}.
\]
It is easy enough to see that $\sum (\alpha^i)^2 = 0$ iff $\sum \alpha^i = 0$ since $GF(2^m)$ is a field of characteristic two -- consequently every other row of $H'$ is redundant. We can remove every other row so that $H_{ij} = (\alpha^{2j + 1})^{i + 1}$, for $i,j$ starting at zero.

Representing elements of $GF(2^m)$ as bit strings of length $m$ produces a parity check matrix with $tm$ rows. Since $2^m \geq w$, this corresponds to $t\ceil{\log_2 w}$ rows (or parity checks).

We now show that $H'$ yields enough information to correct up to $t$ errors. Recall that this is equivalent to any set of up to $2t$ columns of $H'$ being linearly independent. Suppose for the sake of contradiction that $H'v = 0$ for some $v$ such that $\abs{v} \leq 2t$, where we use the full (redundant) matrix $H'$. Then \[\begin{pmatrix}
    \alpha^{j_1} & \alpha^{j_2} & \ldots & \alpha^{j_{\abs{v}}}\\
    (\alpha^{j_1})^2 & (\alpha^{j_2})^2 & \ldots & (\alpha^{j_{\abs{v}}})^2\\
                 &             & \vdots &                   \\
    (\alpha^{j_1})^{2t} & (\alpha^{j_2})^{2t} & \ldots & (\alpha^{j_{\abs{v}}})^{2t}\\
\end{pmatrix} \begin{pmatrix}
1 \\ 1 \\ \vdots \\ 1
\end{pmatrix} = 0\] for all $j_i$ such that $v_{j_i} = 1$.
Since $2t \leq w$ we can truncate this equation to the first $2t$ rows so that the matrix on the left hand side is square. Factoring out a power of $\alpha^{j_i}$ from the $j_i$ column shows that this equation reduces to the determinant of a Vandermonde matrix being equal to zero, which is a contradiction. Since $H$ has the same column independence properties as $H'$, this concludes the proof that $H$ is a parity check matrix for a $t$ error correcting code. Note that it is possible that the true distance of the code given by $\ker H$ is greater than $2t + 1$.
\section{Ordering columns of parity check matrix}\label{appendix:ordering}
In Section~\ref{sec:unfolding} we required that the parity check matrix in question admits an ordering of its columns such that, when unfolded, there are no zero columns between two nonzero columns. Here we provide an explicit ordering satisfying this constraint.
\begin{lemma}
    For any arbitrary binary matrix $H$ that does not contain the zero vector as a column, there exists a matrix $H'$ obtained by reordering columns of $H$ such that $\forall_{k, i}: H'_i \oplus [H'_i \oplus H'_{i - 1}][k] \not = 0$.
\end{lemma}
\begin{proof}
Obtain $H'$ from $H$ by sorting the columns of $H$ in descending order from left to right, where we interpret each column as a binary integer with the highest significance bit at the bottom.

Suppose for the sake of contradiction that there exists some $i, k$ such that $H'_i \oplus [H'_i \oplus H'_{i - 1}][k] = 0$. 
Let the index of the $k-$th nonzero element of $H'_i \oplus H'_{i - 1}$ be called $j$.

Since $H'_i \oplus [H'_i \oplus H'_{i - 1}][k] = 0$, we have $H'_i$ is zero after index $j$.
By the ordering assumption, this implies that $H'_{i - 1}$ is zero after index $j$ as well.

Similarly, since the first $j$ elements of $H'_i \oplus [H'_i \oplus H'_{i - 1}] $, are equal to the first $j$ elements of $H'_i \oplus H'_i \oplus [H'_{i - 1}] = H'_{i - 1}$, we have the first $j$ elements of $H'_{i -1}$ are zero. So $H'_{i - 1}$ is the zero vector.
This contradicts our assumption on $H$.

So sorting the columns in descending order yields an ordering which does not produce the zero vector as an intermediate column.
\end{proof}
\section{Fault Tolerance of Flagged Connected Syndromes}\label{appendix:connected}
We prove that the flagged connected syndrome procedure given in Section~\ref{sec:multi_syn} is fault-tolerant. First, consider the case in which we can perfectly identify all errors on the data that propagated from a syndrome qubit. In this case, it is clear that the entire method of repeated syndrome extraction is fault-tolerant, since after measuring the flags we can track the effect that these perfectly identified errors have on the syndromes and adjust syndrome bits as necessary.

Our construction does not allow us to perfectly identify propagating errors and their effects on the data. However, we can reduce to the previous case easily. If $k \leq t$ errors are suffered on the syndromes, we identify a correction that leaves at most $k$ errors on the data. This is equivalent to perfectly identifying the error on the data that propagated from the syndrome and also suffering $k$ unknown errors on the data. So flagging connected syndromes and only applying the corrections implied by the flags at the end does not impact fault-tolerance.
\section{Reducing Effective Stabilizer Weight by \texorpdfstring{$\frac 1 2$}{1/2}}\label{appendix:double}
In \cite{4}, the authors make the observation that placing two data CNOTs in each region uniquely identified by a flag pattern is sufficient to ensure fault-tolerance. To show this, the correction rule is changed slightly so that instead of applying a correction anywhere in the flagged region, the correction must be applied between the pair of data CNOTs. We can apply this optimization to our construction as well, when we repeat the parity check matrix at least $2t + 1$ times. It might also be possible to apply this optimization when repeating fewer than $2t + 1$ times, but proving this requires stronger characterizations of the transition columns.

In the case where we use $(t + 1)^2$ repetitions (as in Section~\ref{sec:decoding}) and apply the correction associated with the most frequent flag subpattern, the modification is obvious. If the BCH decoder gives $\{e_i\}$ as the correction, where $e_i$ corresponds to applying a correction in a certain region, it is enough to simply assume the correction falls between the pair of data CNOTs in that region. 

The procedure is similar for the case in which we use fewer repetitions (Section~\ref{sec:general_gadgets}). For a given flag pattern $P$, the algorithm outlined produces a correction of minimal weight which has the same flag pattern. The only difference when we add a data CNOT as a pair to each previous data CNOT is that the set of minimal weight corrections may include non-fault-tolerant corrections. However, the correction corresponding to the case where each component of the correction is between the new pair of data CNOTs is fault-tolerant. So we just need some rules for identifying this correction. It is clear that if every component correction comes after an odd number of data CNOTs the total correction will consist of corrections falling between pairs of data CNOTs and will hence be fault-tolerant. Such a correction will always exist in the set of minimal weight corrections with syndrome $P$, since shifting a component correction past at most one data CNOT within the same region of flags does not change the weight or the syndrome.

It is important to note that it may not be possible to replace each data CNOT/qubit with two data CNOTs/qubits in general. It is only possible in our construction because of the fact that errors with the same flag pattern differ in a very specific way -- namely that if two errors $e_1, e_2$ have the same flag pattern each physical error on the syndrome can be paired with another physical error on the syndrome from either $e_1$ or $e_2$ with at most one data CNOT in between (or two data CNOTs in between when replacing each data CNOT by two data CNOTs), as proven in Lemma~\ref{lem:same_flag_pattern}. 

\end{document}

%% file: figures/hamming7norep.tex
\hspace{3.2em}
\begin{tabular}{c}
    \vspace{-.6em}\\
    \Qcircuit @C=1em @R=0.3em @!R{
    \lstick{FLAG1} & \targ    & \qw       & \qw      & \qw       & \qw      & \qw      & \qw       & \targ    & \qw       & \qw      & \targ    & \qw       & \qw      & \targ    & \qw       & \qw      & \qw       & \qw      & \qw  & \measureD{Z}\\
\lstick{FLAG2} & \qw      & \qw       & \targ    & \qw       & \qw      & \targ    & \qw       & \qw      & \qw       & \targ    & \qw      & \qw       & \qw      & \qw      & \qw       & \qw      & \qw       & \qw      & \targ  & \measureD{Z} \\
\lstick{FLAG3} & \qw      & \qw       & \qw      & \qw       & \targ    & \qw      & \qw       & \qw      & \qw       & \qw      & \qw      & \qw       & \targ    & \qw      & \qw       & \targ    & \qw       & \targ    & \qw  &\measureD{Z}   \\
\lstick{SYN} & \ctrl{-3} & \ctrl{1} & \ctrl{-2} & \ctrl{1} & \ctrl{-1} & \ctrl{-2} & \ctrl{1} & \ctrl{-3} & \ctrl{1} & \ctrl{-2} & \ctrl{-3} & \ctrl{1} & \ctrl{-1} & \ctrl{-3} & \ctrl{1} & \ctrl{-1} & \ctrl{1} & \ctrl{-1} & \ctrl{-2} & \measureD{X}\\
\lstick{DATA} & \qw      & \qw       & \qw      & \qw       & \qw      & \qw      & \qw       & \qw      & \qw       & \qw      & \qw      & \qw       & \qw      & \qw      & \qw       & \qw      & \qw       & \qw      & \qw  & \qw
}
\vspace{.2em}     
\\  
\end{tabular}

%% file: figures/fivequbit_flag.tex
\hspace{3.2em}
\begin{tabular}{c}
    \vspace{-.6em}\\
    \Qcircuit @C=1em @R=0.3em @!R{
    \lstick{FLAG1} & \qw      & \qw      & \qw      & \qw       & \qw      & \qw      & \qw      & \qw       & \targ    & \qw      & \qw      & \qw      & \qw      & \qw      & \qw       & \qw      & \qw      & \qw      & \qw       & \targ    & \qw      & \qw     & \measureD{Z} \\
\lstick{FLAG2} & \qw      & \qw      & \qw      & \qw       & \targ    & \qw      & \qw      & \qw       & \qw      & \targ    & \qw      & \qw      & \qw      & \qw      & \qw       & \qw      & \qw      & \qw      & \qw       & \qw      & \qw      & \qw     & \measureD{Z} \\
\lstick{FLAG3} & \targ    & \qw      & \qw      & \qw       & \qw      & \qw      & \qw      & \qw       & \qw      & \qw      & \qw      & \qw      & \qw      & \qw      & \qw       & \targ    & \qw      & \qw      & \qw       & \qw      & \qw      & \qw     & \measureD{Z} \\
\lstick{FLAG4} & \qw      & \qw      & \qw      & \qw       & \qw      & \qw      & \qw      & \qw       & \qw      & \qw      & \targ    & \qw      & \qw      & \qw      & \qw       & \qw      & \qw      & \qw      & \qw       & \qw      & \targ    & \qw     & \measureD{Z} \\
\lstick{FLAG5} & \qw      & \qw      & \qw      & \qw       & \qw      & \targ    & \qw      & \qw       & \qw      & \qw      & \qw      & \targ    & \qw      & \qw      & \qw       & \qw      & \qw      & \qw      & \qw       & \qw      & \qw      & \qw     & \measureD{Z} \\
\lstick{FLAG6} & \qw      & \targ    & \qw      & \qw       & \qw      & \qw      & \qw      & \qw       & \qw      & \qw      & \qw      & \qw      & \qw      & \qw      & \qw       & \qw      & \targ    & \qw      & \qw       & \qw      & \qw      & \qw     & \measureD{Z} \\
\lstick{FLAG7} & \qw      & \qw      & \qw      & \qw       & \qw      & \qw      & \qw      & \qw       & \qw      & \qw      & \qw      & \qw      & \targ    & \qw      & \qw       & \qw      & \qw      & \qw      & \qw       & \qw      & \qw      & \targ   & \measureD{Z} \\
\lstick{FLAG8} & \qw      & \qw      & \qw      & \qw       & \qw      & \qw      & \targ    & \qw       & \qw      & \qw      & \qw      & \qw      & \qw      & \targ    & \qw       & \qw      & \qw      & \qw      & \qw       & \qw      & \qw      & \qw     & \measureD{Z} \\
\lstick{FLAG9} & \qw      & \qw      & \targ    & \qw       & \qw      & \qw      & \qw      & \qw       & \qw      & \qw      & \qw      & \qw      & \qw      & \qw      & \qw       & \qw      & \qw      & \targ    & \qw       & \qw      & \qw      & \qw & \measureD{Z} \\ 
\lstick{SYN} & \ctrl{-7} & \ctrl{-4} & \ctrl{-1} & \ctrl{1} & \ctrl{-8} & \ctrl{-5} & \ctrl{-2} & \ctrl{2} & \ctrl{-9} & \ctrl{-8} & \ctrl{-6} & \ctrl{-5} & \ctrl{-3} & \ctrl{-2} & \ctrl{3} & \ctrl{-7} & \ctrl{-4} & \ctrl{-1} & \ctrl{4} & \ctrl{-9} & \ctrl{-6} & \ctrl{-3} & \measureD{X}\\
\lstick{DATA} & \qw      & \qw      & \qw      & \targ     & \qw      & \qw      & \qw      &  \qw      & \qw      & \qw      & \qw      & \qw      & \qw      & \qw      & \qw       & \qw      & \qw      & \qw      & \qw       & \qw      & \qw      & \qw & \qw\\
\lstick{DATA} & \qw      & \qw      & \qw      & \qw       & \qw      & \qw      & \qw      & \gate{Z}  & \qw      & \qw      & \qw      & \qw      & \qw      & \qw      & \qw       & \qw      & \qw      & \qw      & \qw       & \qw      & \qw      & \qw    & \qw \\
\lstick{DATA} & \qw      & \qw      & \qw      & \qw       & \qw      & \qw      & \qw      & \qw       & \qw      & \qw      & \qw      & \qw      & \qw      & \qw      & \gate Z   & \qw      & \qw      & \qw      & \qw       & \qw      & \qw      & \qw     & \qw\\
\lstick{DATA} & \qw      & \qw      & \qw      & \qw       & \qw      & \qw      & \qw      & \qw       & \qw      & \qw      & \qw      & \qw      & \qw      & \qw      & \qw       & \qw      & \qw      & \qw      & \targ     & \qw      & \qw      & \qw & \qw 
}
\vspace{.2em}    
\\  
\end{tabular}

%% file: figures/syndrome_extraction.tex
\hspace{3.2em}
\begin{tabular}{c}
    \vspace{-.6em}
    \Qcircuit @C=1em @R=0.3em @!R{
    \lstick{Syndrome_1}&     &     &                            &     & \qw & \qw & \ctrl{2}                   & \measureD{X}  \\
    \lstick{Syndrome_2}& \qw & \qw & \ctrl{1}                   & \qw & \measureD{X}  \\
    & \qw & \qw & \multigate{6}{\mathcal{P}} & \qw & \qw & \qw & \multigate{6}{\mathcal{P'}} & \qw & \qw \\
    & \qw & \qw & \ghost{\mathcal{F}}        & \qw & \qw & \qw & \ghost{\mathcal{F}}        & \qw & \qw \\
    & \qw & \qw & \ghost{\mathcal{F}}        & \qw & \qw & \qw & \ghost{\mathcal{F}}        & \qw & \qw \\
    &     &     &                            &     & \large \vdots\\
    & \qw & \qw & \ghost{\mathcal{F}}        & \qw & \qw & \qw & \ghost{\mathcal{F}}        & \qw & \qw \\
    & \qw & \qw & \ghost{\mathcal{F}}        & \qw & \qw & \qw & \ghost{\mathcal{F}}        & \qw & \qw \\
    & \qw & \qw & \ghost{\mathcal{F}}        & \qw & \qw & \qw & \ghost{\mathcal{F}}        & \qw & \qw 
    \inputgroupv{4}{8}{4em}{3.3em}{\hspace{-5em}Data}\\
    &
}
\vspace{.2em}
\end{tabular}

%% file: figures/connected_syndrome_extraction.tex
\hspace{3.2em}
\begin{tabular}{c}
    \vspace{-.6em}
    \Qcircuit @C=1em @R=0.3em @!R{
    \lstick{Syndrome_1}&     &     &                            & \targ    & \qw & \qw & \ctrl{2}                   & \measureD{X}  \\
    \lstick{Syndrome_2}& \qw & \qw & \ctrl{1}                   & \ctrl{-1} & \measureD{X}  \\
    & \qw & \qw & \multigate{6}{\mathcal{P}} & \qw & \qw & \qw & \multigate{6}{\mathcal{P'}} & \qw & \qw \\
    & \qw & \qw & \ghost{\mathcal{F}}        & \qw & \qw & \qw & \ghost{\mathcal{P'}}        & \qw & \qw \\
    & \qw & \qw & \ghost{\mathcal{F}}        & \qw & \qw & \qw & \ghost{\mathcal{P'}}        & \qw & \qw \\
    &     &     &                            &     & \large \vdots\\
    & \qw & \qw & \ghost{\mathcal{F}}        & \qw & \qw & \qw & \ghost{\mathcal{P'}}        & \qw & \qw \\
    & \qw & \qw & \ghost{\mathcal{F}}        & \qw & \qw & \qw & \ghost{\mathcal{P'}}        & \qw & \qw \\
    & \qw & \qw & \ghost{\mathcal{F}}        & \qw & \qw & \qw & \ghost{\mathcal{P'}}        & \qw & \qw 
    \inputgroupv{4}{8}{4em}{3.3em}{\hspace{-5em}Data}\\
    &
}
\vspace{.2em}
\end{tabular}

%% file: figures/connected_syndrome_extraction_flagged.tex
\hspace{3.2em}
\begin{tabular}{c}
    \vspace{-.6em}
    \Qcircuit @C=1em @R=0.3em @!R{
    \lstick{Flag}&\targ & \qw & \cdots                     &           & \targ     & \qw & \cdots & & \measureD{Z}  \\
    \lstick{Syndrome_1}&     &     &                             & \targ     & \ctrl{-1} & \qw &\ctrl{2}                   & \measureD{X}  \\
    \lstick{Syndrome_2}& \ctrl{-2} & \qw & \ctrl{1}                   & \ctrl{-1} & \measureD{X}  \\
    \lstick{}& \qw & \qw & \multigate{6}{\mathcal{P}} & \qw & \qw & \qw & \multigate{6}{\mathcal{P'}} & \qw & \qw \\
    \lstick{}& \qw & \qw & \ghost{\mathcal{F}}        & \qw & \qw & \qw & \ghost{\mathcal{P'}}        & \qw & \qw \\
    \lstick{}& \qw & \qw & \ghost{\mathcal{F}}        & \qw & \qw & \qw & \ghost{\mathcal{P'}}        & \qw & \qw \\
    \lstick{}&     &     &                            &     & \large \vdots\\
    \lstick{}& \qw & \qw & \ghost{\mathcal{F}}        & \qw & \qw & \qw & \ghost{\mathcal{P'}}        & \qw & \qw \\
    \lstick{}& \qw & \qw & \ghost{\mathcal{F}}        & \qw & \qw & \qw & \ghost{\mathcal{P'}}        & \qw & \qw \\
    \lstick{}& \qw & \qw & \ghost{\mathcal{F}}        & \qw & \qw & \qw & \ghost{\mathcal{P'}}        & \qw & \qw 
    \inputgroupv{5}{9}{4em}{3.3em}{\hspace{-5em}Data}\\
    &
}
\vspace{.2em}
\end{tabular}

%% file: figures/5qubit2flag.tex
\hspace{3.5em}
\begin{tabular}{c}
    \vspace{-.6em}\\
    \Qcircuit @C=1em @R=0.3em @!R{
\lstick{FLAG1} & \qw      & \qw       & \qw      & \qw       & \targ    & \qw       & \qw       & \qw      & \qw       & \targ & \measureD{Z}\\
\lstick{FLAG2} & \targ    & \qw       & \targ    & \qw       & \qw      & \qw       & \qw       & \targ    & \qw       & \qw & \measureD{Z}\\   
\lstick{SYN}   & \ctrl{-1} & \ctrl{1} & \ctrl{-1} & \ctrl{2} & \ctrl{-2} & \ctrl{3} & \ctrl{4} & \ctrl{-1} & \ctrl{5} & \ctrl{-2} & \measureD{X} \\
\lstick{DATA1}  & \qw      & \targ     & \qw      & \qw       & \qw      & \qw       & \qw       & \qw      & \qw       & \qw\\
\lstick{DATA2}  & \qw      & \qw       & \qw      & \targ     & \qw      & \qw       & \qw       & \qw      & \qw       & \qw\\
\lstick{DATA3}  & \qw      & \qw       & \qw      & \qw       & \qw      & \targ     & \qw       & \qw      & \qw       & \qw\\
\lstick{DATA4}  & \qw      & \qw       & \qw      & \qw       & \qw      & \qw       & \targ     & \qw      & \qw       & \qw\\
\lstick{DATA5}  & \qw      & \qw       & \qw      & \qw       & \qw      & \qw       & \qw       & \qw      & \targ     & \qw
}
\vspace{.2em}    
\\  
\end{tabular}